\documentclass[11pt]{article}

\usepackage{graphicx}
\usepackage{amssymb,amsmath,amsthm}
\usepackage{fullpage}
\usepackage[left]{lineno}
\usepackage{adjustbox}
\usepackage[utf8]{inputenc}
\usepackage[linesnumbered,boxed,ruled]{algorithm2e}

\newcommand{\comment}[1]{} 

\newtheorem{theorem}{Theorem}

\newtheorem{lemma}[theorem]{Lemma}

\newtheorem{prop}[theorem]{Proposition}
\newtheorem{conjecture}{Conjecture}

\newcounter{quote}

\newcommand\blfootnote[1]{%
  \begingroup
  \renewcommand\thefootnote{}\footnote{#1}%
  \addtocounter{footnote}{-1}%
  \endgroup
}

\usepackage{amsfonts}

\usepackage{epsfig}
\usepackage{graphicx}
\usepackage{subcaption}

\usepackage[inline]{enumitem}

\newcommand{\ra}[2]{\ensuremath{\overrightarrow{#1 #2}}}
\newcommand{\la}[2]{\ensuremath{\overleftarrow{#1} \hspace*{-2pt} #2}}

\newcommand{\crs}[1]{\overline{\operatorname{cr}}(#1)}

\title{Counting the Number of Crossings in Geometric Graphs}
\author{Frank Duque\thanks{Instituto de Matemáticas UNAM and Instituto de Física UASLP, Mexico; { \tt rodrigo.duque@udea.edu.co}; partially supported by FORDECYT 265667 (Mexico).} 
\and Ruy Fabila-Monroy\thanks{Departamento de Matem\'aticas CINVESTAV, Mexico. {\tt ruyfabila@math.cinvestav.edu.mx, cmhidalgo@math.cinvestav.mx;} partially supported
by Conacyt of Mexico grant 253261.}
\and César Hern\'andez-V\'elez\thanks{Facultad de Ciencias UASLP; {\tt cesar.velez@uaslp.mx}}
\and Carlos Hidalgo-Toscano\footnotemark[2] }

\begin{document}


\maketitle

\blfootnote{\begin{minipage}[l]{0.3\textwidth} \includegraphics[trim=10cm 6cm 10cm 5cm,clip,scale=0.15]{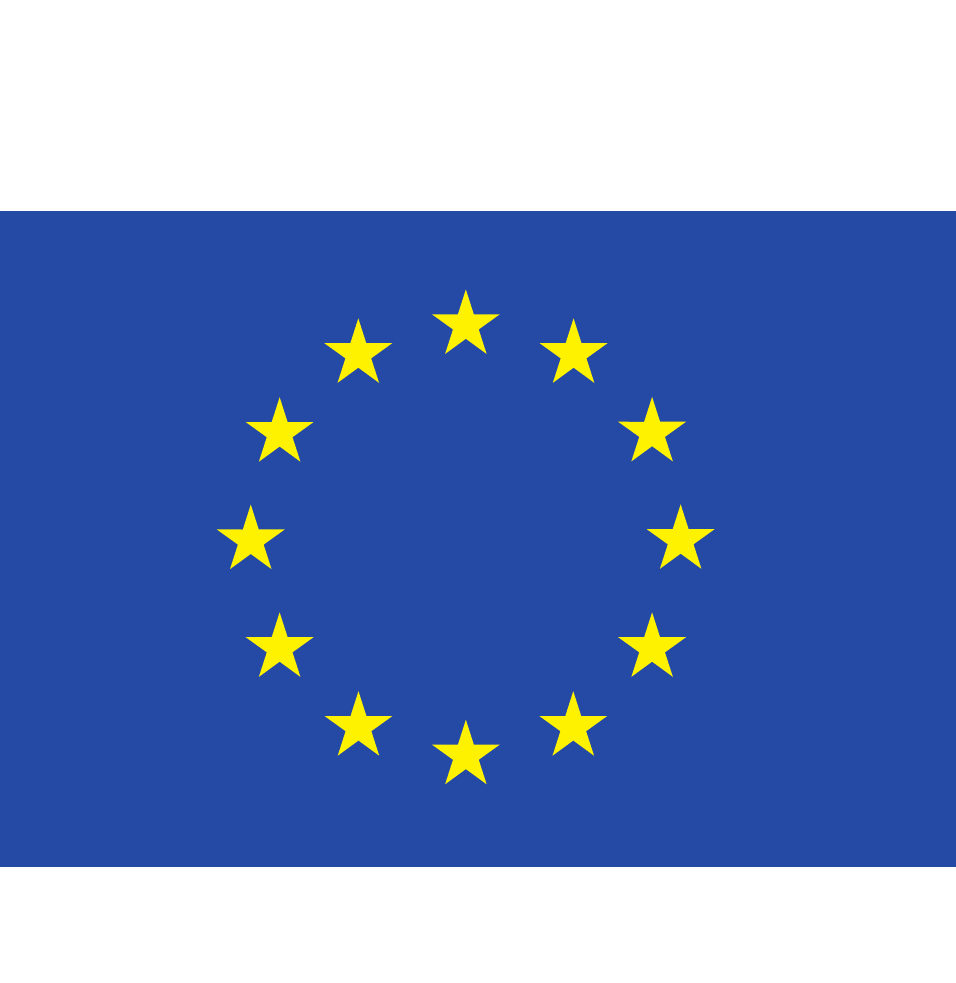} \end{minipage}  \hspace{-3.5cm} \begin{minipage}[l][1cm]{0.82\textwidth}
 	  This project has received funding from the European Union's Horizon 2020 research and innovation programme under the Marie Sk\l{}odowska-Curie grant agreement No 734922.
 	\end{minipage}}

\begin{abstract} 
A geometric graph is a graph whose vertices are points in general position in the plane and its edges are straight line
segments joining these points. In this paper we give an $O(n^2 \log n)$ time algorithm to compute the number of pairs
of edges that cross in a geometric graph on $n$ vertices. For 
layered graphs and convex geometric graphs the algorithm takes $O(n^2)$ time. 
\end{abstract}

\section{Introduction}\label{sec:intro}

A \emph{geometric graph} is a graph  whose vertices are points 
in general position in the plane; and its edges are straight line segments
joining these points. A pair of edges of a geometric graph \emph{cross} if they intersect in their interior;
the \emph{number of crossings} of a geometric graph is the number of pairs of its edges that cross. 
Let $G:=(V,E)$ be a geometric graph on $n$ vertices; and let $H$ be a  graph. 
If $G$ and $H$ are isomorphic as graphs then we say that $G$ is a \emph{rectilinear
drawing} of $H$.
The \emph{rectilinear crossing number} of $H$ is the minimum number of crossings
that appear in all of its rectilinear drawings. We abuse notation and use 
$\crs{H}$ and $\crs{G}$ to denote the rectilinear crossing number of $H$ and
the number of crossings of $G$, respectively. 

Computing the rectilinear crossing number of the complete graph $K_n$ on $n$ vertices is an important
and well known problem in Combinatorial Geometry. The current best bounds on $\crs{K_n}$ are
\[ 0.379972 \binom{n}{4} < \crs {K_n}  < 0.380473\binom{n}{4}+\Theta(n^3). \]
The lower bound is due to {\'A}brego, Fern{\'a}ndez-Merchant, Lea{\~n}os and Salazar~\cite{lower}.
The upper bound is due to Fabila-Monroy and L{\'o}pez~\cite{upper_old}. In an upcoming paper, Aichholzer, Duque,
Fabila-Monroy, Garc\'ia-Quintero and Hidalgo-Toscano~\cite{upper_new} have further improved the upper bound to 
\[\crs{K_n}  < 0.3804491\binom{n}{4}+\Theta(n^3). \]
 For more information on crossing numbers (rectilinear or other variants) 
we recommend the survey by Schaefer~\cite{cr_survey}.

A notable property of the improvements of~\cite{upper_old,upper_new} on the upper bound of $\crs{K_n}$
is that they rely on finding rectilinear drawings of $K_n$ with few crossings, for some specific values
of $n$; this is done via heuristics that take a rectilinear drawing of $K_n$ and move its vertices in various ways with the
aim of decreasing the number of crossings.
In this approach it is instrumental that the computation of the number of crossings
is done as fast as possible.  In this paper we present an algorithm
 to compute $\crs{G}$ in $O(n^2\log n)$ time. We hope that our algorithm will pave the way for finding
new upper bounds on the rectilinear crossing number of various classes of graphs.

The current best algorithm for counting the number of intersections 
among a set of $m$ line segments in the plane runs in $O(m^{4/3} \log^{1/3}m)$ time; it was given by Chazelle~\cite{chazelle1993}.
This provides an $O(|E|^{4/3} \log^{1/3} |E|)$ time algorithm for computing $\crs{G}$; this running time can be as high as $\Theta(n^{8/3} \log^{1/3} n)$
when $G$ has a quadratic number of edges.
The running time can be improved for some classes of geometric graphs:
when $G$ is a complete geometric graph, Rote, Woeginger, Zhu, and Wang~\cite{rote1991} give an $O(n^2)$ time algorithm for computing $\crs{G}$; and
Waddle and Malhotra~\cite{waddle1999}  give an $O(|E| \log |E|)$ algorithm for computing $\crs{G}$ 
when  $G$ is a bilayered graph. 
 For layered graphs and convex geometric graphs
 a slight modification of our algorithm runs in $O(n^2)$ time.

\section{The algorithm}\label{sec:algorithm}

In what follows let $G:=(V,E)$ be a geometric graph on $n$ vertices. We make some general position
assumptions: no two vertices of $G$ have the same $x$-coordinate nor the same $y$-coordinate; and no two edges of $G$ are parallel.
In \cite{upper_old}, the authors give an $O(n^2)$ time algorithm for computing $\crs{G}$ when $G$ is a
complete geometric graph; the authors define 
``patterns'' on the set of vertices of $G$; these patterns can be computed in $O(n^2)$ time and
$\crs{G}$ depends on the number of these patterns. We follow a similar approach. 

 Let $p$ and $q$ be two points in the plane. Let $\ra{p}{q}$ be the ray with apex $p$ and that passes through $q$; let $\la{p q}$ be the ray
 with apex $p$, parallel to $\ra{p}{q}$, and with opposite direction to $\ra{p}{q}$.
 Let $(u,v,e)$ be a triple where $u$ and $v$ are a pair of adjacent vertices of $G$, and $e$ is an edge of $G$ not incident to $u$ or $v$.
 We say that $(u,v,e)$ is a pattern of
\begin{itemize}
	\item \textbf{Type $A$)} if $\ra{u}{v}$ intersects $e$; and
	\item \textbf{Type $B$)} if $\la{uv}$ intersects $e$.
\end{itemize}
See Figure~\ref{fig:patterns}. 

Consider a pair of non-incident edges of $G$. This pair of edges provides patterns of type $A$ and type
$B$ according to the following three cases. If the pair of edges is crossing then they
provide exactly four patterns of type $A$ and zero patterns of type $B$;
if the pair of edges is non-crossing and the line containing one of them intersects the interior
of the other edge, then they provide exactly one pattern of type $A$ 
and one pattern of type $B$; and if the pair of edges is non-crossing and the intersection
point of the two supporting lines of these edges lies outside the union of these edges,
then they do not provide patterns of type $A$ nor patterns of type $B$. 
See~Figure~\ref{fig:patterns}. Let $A(G)$ and $B(G)$ be the number of patterns of type $A$ and type $B$ defined by $G$, respectively; 
we have the following. 
 \begin{figure}
 	\centering
 	\includegraphics[width=0.33\textwidth]{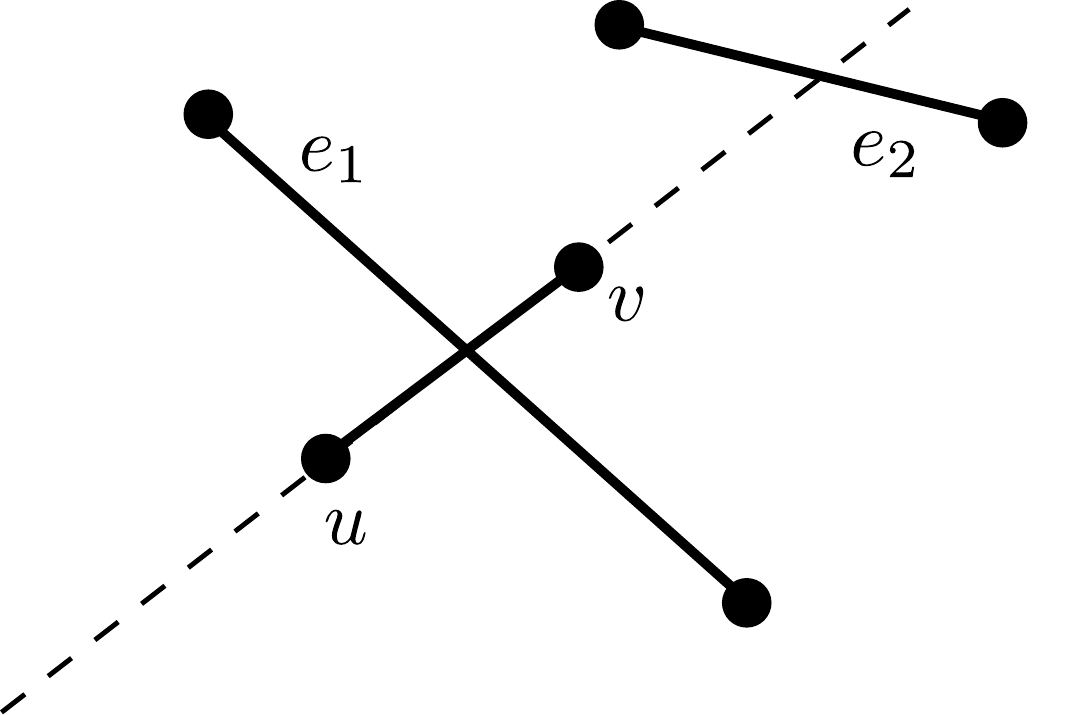}
 	\caption{$(u,v,e_1)$, $(v,u,e_1)$, and $(u, v, e_2)$ are of type $A$; $(v,u,e_2)$ is of type $B$; and
 	neither the endpoints of $e_1$ and the edge $e_2$, nor the endpoints of $e_2$ and the edge $e_1$ produce
 	patterns of type $A$ or type $B$.}
 	\label{fig:patterns}
 \end{figure}

\begin{prop}\label{lem:patterns}
  \[\crs{G}=\frac{A(G)-B(G)}{4}.\]
\end{prop}
\qed


Before proceeding we give some definitions.
Let $u$ and $v$ be vertices of $G$.  Let $B(u,v)$ and $A(u,v)$ be the number of edges
of $G$ that intersect $\ra{u}{v}$ and  $\la{uv}$, respectively. Note that 
\[A(G)=\sum_{u \in V} \sum_{v \in N(u)} A(u,v) \text{ and } B(G)=\sum_{u \in V} \sum_{v \in N(u)} B(u,v).\]
Let $\textsc{N}_\textsc{Left}(v, uv)$ be set of neighbors of $v$ to the left of the directed line from $u$ to
$v$; and let $\textsc{N}_\textsc{Right}(v, uv)$ be set of neighbors of $v$ to the right of the directed line from $u$ to
$v$.  See Figure~\ref{fig:vw}. Let $\overleftarrow{u}$ and  $\overrightarrow{u}$  be the open horizontal rays with apex $u$ 
that go left and  right, respectively.
Finally, let $E(u)$ be the set of edges of $G$ that are incident to $u$.
Thus, $|\overleftarrow{u} \cap E|$ and $|\overrightarrow{u} \cap E|$ are equal to the number of edges of $G$ that cross  $\overleftarrow{u}$ and $\overrightarrow{u}$,
respectively;
and  $|\overleftarrow{u} \cap E(v)|$ and $|\overrightarrow{u} \cap E(v)|$ are equal to the number of edges of $G$ incident to $v$ that 
cross $\overleftarrow{u}$ and $\overrightarrow{u}$, respectively. 
\begin{figure} [htb]
      \centering
      \includegraphics[scale=0.8]{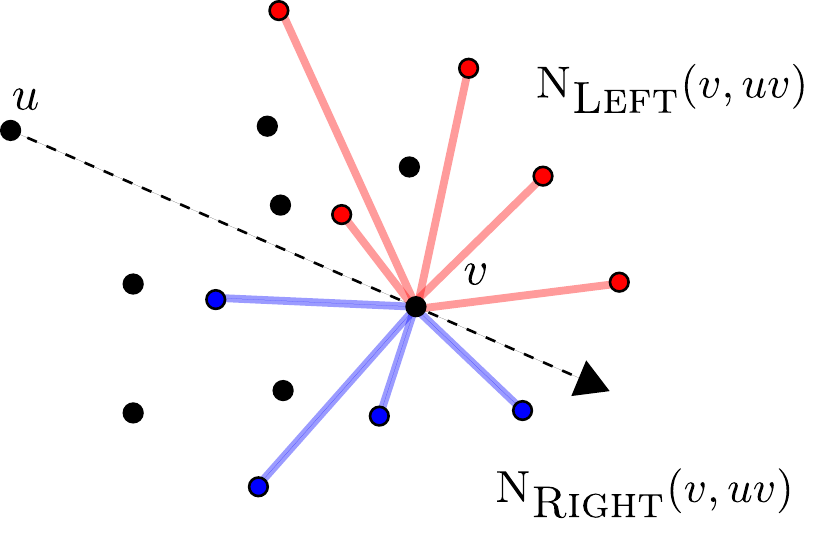}
    \caption{An illustration of $\textsc{N}_\textsc{Left}(v, uv)=5$ and $\textsc{N}_\textsc{Right}(v, uv)=4$}
    \label{fig:vw}
\end{figure}

Consider the following well known duality transform. Let $p:=(m,b)$ be a point; the dual of $p$ is the line
$p^\ast:y=mx-b$. Let $\ell:y=mx+b$ be a line; the dual of $\ell$ is the point $\ell^\ast:=(m,-b)$.
This  duality transform preserves various geometric relationships between points and lines. In particular, it
preserves point-line incidences. It also preserves above-below relationships between lines and points: if point $p$ is above (resp. below)
line $\ell$, then point $\ell^{\ast}$ is above (resp. below) line $p^{\ast}$. We compute the line arrangement $\mathcal{L}$ of the dual lines
of the points in $V$. This can be done in $O(n^2)$ time; see for example Chapter~8 of \cite{cgbook}.
Afterwards, we compute $A(G)$ and $B(G)$ in the following four steps.

\begin{lemma}[Step 1]\label{lem:step1}
In $O(n^2)$ time we can compute, for all $u\in V$, the counterclockwise order of the vertices in
$V\setminus \{u\}$ by angle  around $u$.
\end{lemma}
\begin{proof}
Let $u$ be a vertex of $G$ and let $v_1,\dots, v_{n-1}$ be the vertices  $V \setminus\{ u\}$. For $1 \le i \le n-1$
be $\ell_i$ be the line through $u$ and $v_i$. Since all the lines $\ell_i$ pass through $u$, all
the points $\ell_i^\ast$ are contained in the line $u^\ast$. Moreover, the order by slope
of the lines $\ell_1,\dots,\ell_{n-1}$ coincides with the order  of the points  $\ell_1^{\ast},\dots,\ell_{n-1}^{\ast}$ 
from left to right along the line $u^{\ast}$. This order can be obtained from $\mathcal{L}$ in $O(n)$ time. 
From the slope order of the lines $\ell_i$ it is
straightforward to compute in $O(n)$ time the counterclockwise order of the vertices in
$V \setminus \{u\}$ by angle  around $u$.
\end{proof}

%

\begin{lemma}[Step 2]\label{lem:step2}
 In $O(n^2)$ time we can compute, for all pairs of vertices $u,v \in G$, $|\textsc{N}_\textsc{Left}(v, uv)|$ and $|\textsc{N}_\textsc{Right}(v, uv)|.$

\end{lemma}
\begin{proof}
Let $u$ be a vertex of $G$. Let $v_1,\dots,v_{n-1}$ be the vertices of $V \setminus \{u\}$ ordered
so that the points $v_1^{\ast}\cap u^{\ast}, v_2^{\ast}\cap u^{\ast},\dots, v_{n-1}^{\ast}\cap u^{\ast}$
are encountered in this order when traversing the line $u^{\ast}$ from left to right.
If the $x$-coordinate of $u$ is less than the $x$-coordinate of $v_i$,
then the set $\textsc{N}_\textsc{Left}(v_i, uv_i)$ (resp. $\textsc{N}_\textsc{Right}(v_i, uv_i)$)
corresponds to the neighbors $v_j$ of $v_i$ such that the  lines $v_j^{\ast}$ are below (resp. above)
the point $v_i^{\ast}\cap u^{\ast}$. If the $x$-coordinate of $u$ is greater than the $x$-coordinate of $v_i$,
then the set $\textsc{N}_\textsc{Left}(v_i, uv_i)$ (resp. $\textsc{N}_\textsc{Right}(v_i, uv_i)$)
corresponds to the neighbors $v_j$ of $v_i$ such that  the lines $v_j^{\ast}$ are above (resp. below)
the point $v_i^{\ast}\cap u^{\ast}$. We compute $\textsc{N}_\textsc{Left}(v_1, uv_1)$ and $\textsc{N}_\textsc{Right}(v_1, uv_1)$
in linear time. Afterwards, iteratively for $j=2,\dots,n-1$ from $\textsc{N}_\textsc{Left}(v_{j-1}, uv_{j-1})$ and 
$\textsc{N}_\textsc{Right}(v_{j-1}, uv_{j-1})$ we compute  $\textsc{N}_\textsc{Left}(v_{j}, uv_{j})$ and 
$\textsc{N}_\textsc{Right}(v_{j}, uv_{j})$ in constant time, respectively. We store the cardinality of these sets.
Thus, in $O(n)$ time we compute  $|\textsc{N}_\textsc{Left}(v_{j}, uv_{j})|$ and 
$|\textsc{N}_\textsc{Right}(v_{j}, uv_{j})|$ for all $v_j$. By iterating over all vertices $u \in G$ 
the result follows.
\end{proof}

%
%
\begin{figure} [htb]
      \centering
      \includegraphics[scale=0.8]{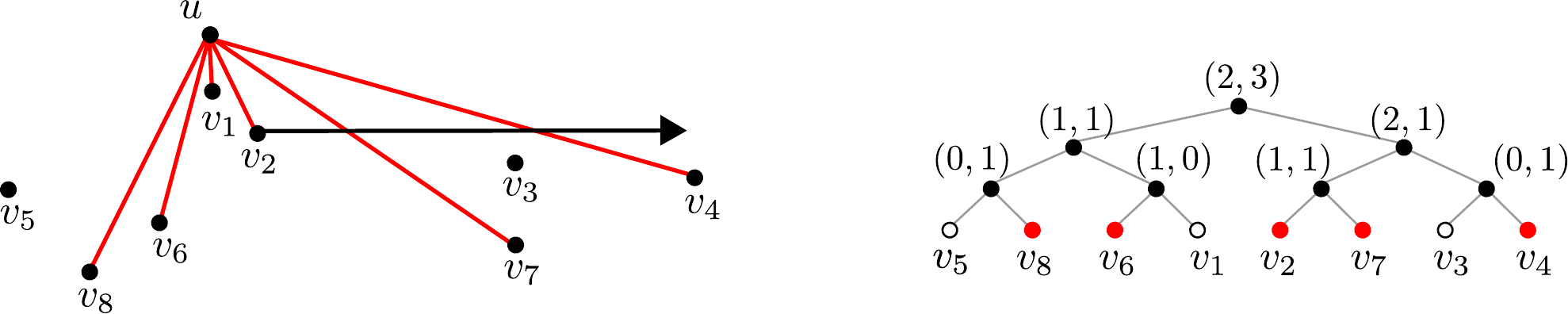}
    \caption{The information stored on $T$ at line $1$ of \textsc{HorizontalRayCrossings} for $i=2$}
    \label{fig:Ti}
\end{figure}

\begin{lemma}[Step 3]
In $O(n^2 \log n)$ time, we can compute the sets of values
\[\{|\overleftarrow{v} \cap E|:v \in V\} \text{ and } \{|\overrightarrow{v} \cap E|:v \in V\}.\]
 \end{lemma}
 \begin{proof}
 We show how to compute $\{|\overrightarrow{v} \cap E|:v \in V\}$; the computation
 of $\{|\overleftarrow{v} \cap E|:v \in V\}$ is analogous.
  We begin by sorting the vertices of $G$ by decreasing $y$-coordinate. For every vertex $u \in V$ 
  we compute the set of values 
  \[\left \{ |\overrightarrow{v} \cap E(u)| : v \in V \text{ and } v \text{ is below } u \right \},\]
  as follows. 
  
  Let $v_1, \dots v_m$ be the vertices of $G$ that lie below $u$ in decreasing order of their
  $y$-coordinate. In $O(n)$ time we construct a binary tree, $T$, of minimum height and whose leaves are $v_1,\dots,v_m$.
The left to right order of $v_1,\dots,v_m$ in $T$ coincides with their counterclockwise order around $u$ starting from 
the first vertex  encountered when rotating $\overleftarrow{u}$ counterclockwise around $u$.
We mark all the $v_i$ that are adjacent to $u$. At each interior node $w$
of $T$ we store in two fields \textsc{MarkedLeft}($w$) and \textsc{MarkedRight}($w$), the number of 
marked leaves in the left and right subtrees of $w$, respectively. By proceeding by increasing height, these values
can be computed in $O(n)$ time for every vertex of $T$.

We execute the procedure \textsc{HorizontalRayCrossings($u$)} whose pseudocode
is described below.
Suppose that we are at the $i$-th execution of the \emph{for} of line $1$. 
We have that: a vertex $v_j$ is marked if and only if $j \ge i$
and $v_j$ is adjacent to $u$; and for every node $w \in T$,  \textsc{MarkedLeft}($w$) and \textsc{MarkedRight}($w$), store the number of 
currently marked leaves in the left and right subtrees of $w$, respectively. Lines $4$, $11$ and $15$
ensure that these invariants are kept at every execution of line $1$. 
In this iteration we compute (and store) the value of
$|\overrightarrow{v_i} \cap E(u)|$.  
Note that an edge $uv_j$ intersects $\overrightarrow{v_i}$ if and only if $v_j$ comes
after $v_i$ in the counterclockwise order around $u$ and $j > i$. 
The first condition is equivalent to $v_j$ being to the right of $v_i$ in $T$;
this happens if and only if there exists a common  ancestor $w$ in $T$ of $v_i$ and $v_j$, for which
$v_i$ is in the left subtree of $w$ and $v_j$ is in the right subtree of $w$;
see Figure~\ref{fig:Ti}. The condition of $j$ being greater than $i$ 
is equivalent to $v_j$ being marked. We count the number of nodes $v_j$
that satisfy both conditions by traversing the path from $v_i$ to the root.
We begin by unmarking $v_i$; at each node $w$ in this path we update the fields \textsc{MarkedLeft}($w$) and \textsc{MarkedRight}($w$);
if  $v_i$ is in the left subtree of $w$ we also update  our counting. Having reached the root we store the value of 
$|\overrightarrow{v_i} \cap E(u)|$ in line $20$.
Since $T$
has height $O( \log n)$ the execution of \textsc{HorizontalRayCrossings($u$)} takes $O(n \log n)$ time.
Since \[|\overrightarrow{v} \cap E| = \sum_{u \text{ is above } v} |\overrightarrow{v} \cap E(u)|,\]
the result follows.
\end{proof}
 \begin{procedure}[h]
  \caption{HorizontalRayCrossings($u$)}
   \For{$i \gets 1$ \KwTo $m$}
    {
        $cr \gets 0$ \;
        \If{$v_i$ is a neighbor of $u$}
        {
            Unmark $v_i$\;
        }
        $w \gets v_i$\;
        \While{$w \neq T.\textsc{root}$}
        {
        \eIf{$w$ is the left child of $w_\textsc{parent}$}
        {
            $cr \gets cr+\textsc{MarkedRight}(w_\textsc{parent})$\;
            \If{$v_i$ is a neighbor of $v$}
            {
                $\textsc{MarkedLeft}(w_\textsc{parent}) \gets \textsc{MarkedLeft}(w_\textsc{parent})-1$;
            }
            
        }
        {
            \If{$v_i$ is a neighbor of $u$}
            {
                $\textsc{MarkedRight}(w_\textsc{parent}) \gets \textsc{MarkedRight}(w_\textsc{parent})-1$;
            }
        }
        $w \gets w_\textsc{parent}$\;
        }
        $|\overrightarrow{v_i} \cap E(u)| \gets cr\;$
    }
 \end{procedure}

\begin{lemma}[Step 3] \label{lem:step3}
 In $O(n^2)$ time we can compute $A(u,v)$ and $B(u,v)$ for all pairs of vertices $u,v \in V$.
\end{lemma}
\begin{proof}
 For every $u \in V$ we proceed as follows. Let $v_1,\dots,v_{n-1}$ be the vertices
 of $V\setminus \{u\}$, in the order as they
 are encountered by rotating a horizontal line counterclockwise around $u$. This
 order can be obtained from $\mathcal{L}$ in $O(n)$ time.
 Suppose that the angle from $\overrightarrow{u}$ to $\ra{u}{v_1}$ is less than $\pi$.
 Then the only vertex of  $V\setminus \{u\}$ in the closed wedge bounded by the rays $\overrightarrow{u}$ and 
 $\ra{u}{v_1}$ is $v_1$; and there are no vertices of $V\setminus \{u\}$ in the closed wedge bounded by 
 the rays  $\overleftarrow{u}$ and $\la{uv_1}$. This implies that 
 \[A(u,v_1)=|\overrightarrow{u} \cap E| -|N_\textsc{Right}(v_1,uv_1)| \text{ and } B(u,v_1)=|\overleftarrow{u} \cap E| .\]
 Suppose that the angle from $\overrightarrow{u}$ to $\ra{u}{v_1}$ is greater than $\pi$.
 Then there are no vertices  $V\setminus \{u\}$ in the closed wedge bounded by the rays $\overrightarrow{u}$ and 
 $\la{uv_1}$ ; and the only vertex of  $V\setminus \{u\}$ in the closed wedge bounded by 
 the rays  $\overleftarrow{u}$ and $\ra{u}{v_1}$ is $v_1$. This implies that 
 \[A(u,v_1)=|\overleftarrow{u} \cap E| -|N_\textsc{Right}(v_1,uv_1)| \text{ and } B(u,v_1)=|\overrightarrow{u} \cap E| ;\]
 
 In general, for  $i=2,\dots,n-1$, we have the following. 
 Suppose that the angle from $\ra{u}{v_{i-1}}$ to $\ra{u}{v_{i}}$ is less than $\pi$. 
 Then the only vertices of  $V\setminus \{u\}$ in the closed wedge bounded by the rays $\ra{u}{v_{i-1}}$  and 
 $\ra{u}{v_{i}}$ are $v_{i-1}$ and $v_{i}$; and there are no vertices of $V\setminus \{u\}$ in the closed wedge bounded 
by the rays $\la{uv_{i-1}}$  and 
 $\la{uv_{i}}$ . This implies that
 \[A(u,v_i)=A(u,v_{i-1})+|N_\textsc{Left}(v_{i-1},uv_{i-1})|-|N_\textsc{Right}(v_{i},uv_{i})|\] and  \[B(u,v_{i})=B(u,v_{i-1}). \]
 Suppose that the angle from $\ra{u}{v_{i-1}}$ to $\ra{u}{v_{i}}$ is greater than $\pi$. 
 Then the only vertex of  $V\setminus \{u\}$ in the closed wedge bounded by the rays $\ra{u}{v_{i-1}}$  and 
 $\la{uv_{i}}$ is $v_{i-1}$; and the only vertex of $V\setminus \{u\}$ in the closed wedge bounded
 by the rays $\la{uv_{i-1}}$  and  $\ra{u}{v_{i}}$ is $v_{i-1}$. This implies that
 \[A(u,v_i)=B(u,v_{i-1})-|N_\textsc{Right}(v_{i},uv_{i})|\] and  \[B(u,v_{i})=A(u,v_{i-1})+|N_\textsc{Left}(v_{i-1},uv_{i-1}). \]
 See Figure~\ref{fig:ABi}.
Thus, having computed $A(u,v_{i-1})$ and $B(u,v_{i-1})$ we can compute $A(u,v_i)$ and $B(u,v_i)$
in constant time. The result follows.
\end{proof}

\begin{figure} [htb]
      \centering
      \includegraphics[width=1.0\textwidth]{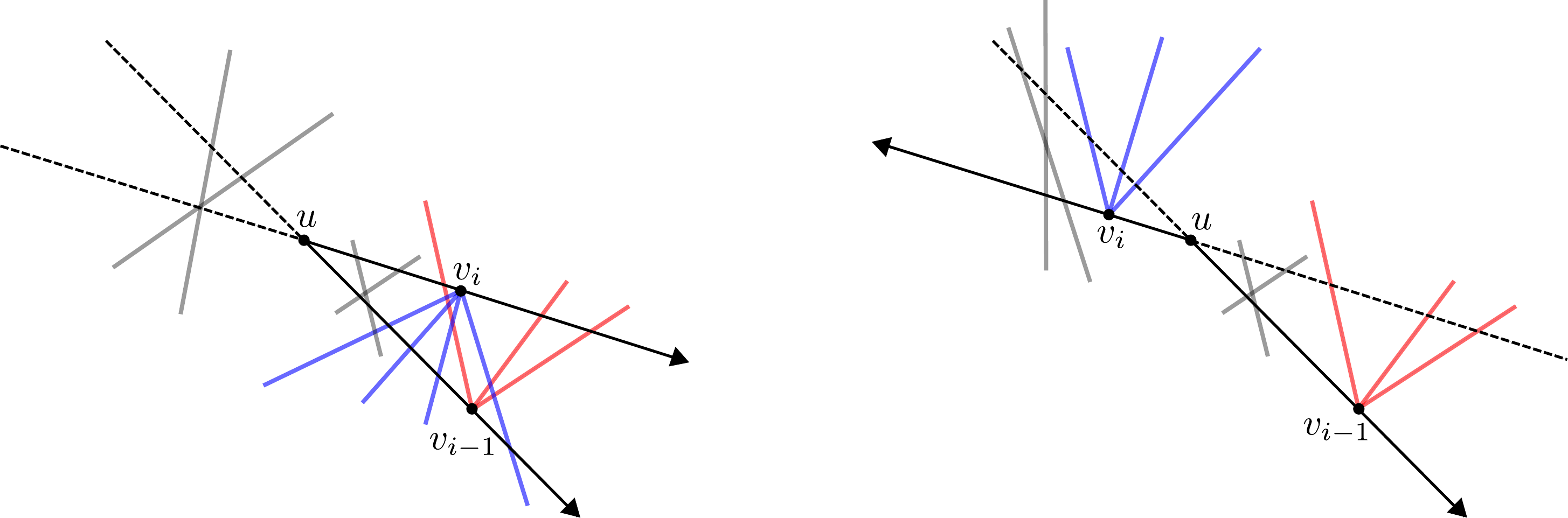}
    \caption{The two cases in computing $A(u,v_i)$ and $B(u,v_i)$ in Step~$4$}
    \label{fig:ABi}
\end{figure}

Summarizing, we have the following result.
\begin{theorem}
 $\crs{G}$ can be computed $O(n^2 \log n)$ time.
\end{theorem}
\qed

\section{Faster algorithms for special classes of graphs} \label{sec:quadratic}

Of the four steps of our algorithm, only Step~$3$ takes superquadratic time. Moreover, 
the numbers $|\overrightarrow{u} \cap E|$ and $|\overleftarrow{u} \cap E|$ are only used as a starting
point for computing $A(u,v)$ and $B(u,v)$ in Lemma~\ref{lem:step3}. If we can compute in $O(n)$ time, for each
vertex $u$ of $G$, the number of edges of $G$ that cross two 
open rays with apex $u$ and  in parallel but opposite directions then
we can compute $\crs{G}$ in $O(n^2)$ time. We pose the following conjecture.
\begin{conjecture}
 For every geometric graph $G$ on $n$ vertices, $\crs{G}$ can be computed
 in $O(n^2)$ time. 
\end{conjecture}

To finalize the paper we mention two instances
for which is the case.

\begin{itemize}

\item \textbf{Convex Geometric Graphs} \\
A \emph{convex geometric graph} is a convex geometric graph whose vertices
are in convex position. Let $G=(V,E)$ be a convex geometric graph on $n$ vertices. 
For every $u$ in $V$, let $\overrightarrow{u}$ and $\overleftarrow{u}$ be two 
parallel open rays with apex $u$, in opposite directions and that do not
intersect the convex hull of $V$. These rays can be found
in constant time by  first computing the convex hull of $V$ (which takes $O(n \log n)$ time).
We have that $|\overrightarrow{u} \cap E|$ and $|\overleftarrow{u} \cap E|$  
are both  equal to zero in this case and $\crs{G}$ can be computed
 in $O(n^2)$ time.

\item  \textbf{Layered Graphs}\\
A \emph{layered graph} is a geometric graph whose vertex set is partitioned
into sets $L_1, \dots, L_r$ called \emph{layers} such that the following holds. 
\begin{itemize}
 \item  The vertices in layer $L_i$ have the same $y$-coordinate $y_i$\footnote{Although this violates our general position assumption, a simple
 perturbation argument suffices to apply our algorithm.} ; 
 \item  $y_1 < y_2 < \dots <y_r$;
 \item vertices in layer $L_i$ are only adjacent
to vertices in layers $L_{i-1}$ and $L_{i+1}$. 
\end{itemize} 
Suppose that $G$ is a layered graph with layers $L_1,\dots,L_r$. 
Let $G_i$ be the subgraph of $G$ induced by $L_{i}$ and $L_{i+1}$. Note that $G_i$ can be regarded
as a convex geometric graph. Since $\crs{G}=\sum_{i=1}^{r-1} \crs{G_i}$, $\crs{G}$ can be computed in $O(n^2)$ time.
Problem~33 of \cite{open} asks whether the 
 number of crossings of layered graphs can be computed in time $o(|E|\log |V|)$.    For 
 layered graphs with $\omega(n^2/\log n)$ edges we provide an affirmative answer to the question
 posed in~\cite{open}.
\end{itemize}
%

\textbf{Acknowledgments.} 
This work was initiated at the \emph{2nd Reunion of Optimization, Mathematics, and Algorithms
(ROMA 2018)}, held in Mexico city, 2018. We thank the anonymous reviewer whose comments substantially
improved the presentation of this paper.



\bibliographystyle{plain}
\bibliography{countingcrossings}

\end{document}